\documentclass[11pt]{article}

\usepackage[a4paper,margin=1in]{geometry}
\usepackage{amsmath,amssymb,amsthm,mathtools}
\usepackage[hidelinks]{hyperref}
\usepackage{microtype}
\usepackage{enumitem}
\usepackage{booktabs}
\usepackage{graphicx}
\usepackage{xcolor}

\theoremstyle{plain}
\newtheorem{theorem}{Theorem}
\newtheorem{proposition}[theorem]{Proposition}
\newtheorem{lemma}[theorem]{Lemma}
\newtheorem{corollary}[theorem]{Corollary}

\theoremstyle{definition}
\newtheorem{definition}[theorem]{Definition}

\newtheorem{problem}[theorem]{Open Problem}

\theoremstyle{remark}
\newtheorem{remark}[theorem]{Remark}

\newcommand{\eps}{\varepsilon}

\newcommand{\bigO}{\mathcal{O}}

\newcommand{\synt}[1]{\mathrm{Synt}(#1)}
\newcommand{\St}{\mathsf{SLT}}
\newcommand{\Lt}{\mathsf{LT}}
\newcommand{\Ap}{\mathsf{Ap}}
\newcommand{\ZG}{\mathbf{ZG}}

\newcommand{\edit}{\mathsf{es}}
\newcommand{\size}{\mathsf{size}}
\newcommand{\loc}{\mathsf{loc}}
\newcommand{\rev}{\mathsf{rv}}
\newcommand{\bedit}{\mathsf{bes}}

\newcommand{\rootnode}{\mathsf{root}}
\newcommand{\dist}{\mathrm{dist}}
\newcommand{\Ham}{\mathrm{Ham}}
\newcommand{\avg}{\mathsf{mean}}

\title{\bfseries Exact Local Annotations for Regular Languages}

\author{Faruk Alpay\textsuperscript{*}\quad Bar{\i}\c{s} Ba\c{s}aran\\[2pt]
\small Department of Computer Engineering, Bah\c{c}e\c{s}ehir University\\
\small Istanbul, Turkey\\
\small \texttt{\{faruk.alpay, baris.basaran\}@bahcesehir.edu.tr}\\[2pt]
\small \textsuperscript{*}Corresponding author: \texttt{alpay@lightcap.ai}}

\date{}

\begin{document}
\maketitle

\begin{abstract}
A regular language can be recognised by a finite monoid, but a locally
checkable explanation of that recognition may have its own geometry.  On the
standard local-certification/proof-labeling model, this paper studies the
edit-stability of exact bounded--arity annotations over a fixed finite label
alphabet for regular word languages under one-symbol substitutions.  The cost
of an edit is not trusted update time; it is the number of annotation cells that
a canonical locally accepted representation must change, the corresponding bit
movement under a fixed label encoding, and the number of local constraints that
must be revalidated.  Equivalently, for each length we ask for a
low--Lipschitz section of a local acceptance relation over the substitution
graph.  The universal theorem is an upper bound: for every morphism
$\eta:\Sigma^{*}\to M$ recognising a regular language $L$, the balanced product
annotation over $M$ gives constant locality, linear size, logarithmic edit
stability, logarithmic revalidation, and constant access to the membership
value.  The matching lower argument proved here is not universal.  It applies
to the committed geometry in which a bounded fan--in product decomposition
exposes an edit-active nontrivial group quotient as projected ordered-product
labels; in that submodel, one substitution changes every quotient label on an
ancestor path.  The gap between these two statements is the main open problem:
whether arbitrary exact bounded--arity annotations can avoid this product-path
transport and achieve constant stability beyond strict locality.  We also show
that annotation-free bounded-window recognition is exactly strict locality, and
we formulate the remaining boundary as a finite obstruction problem.  The
ancillary library contains a small Lean~4 proof kernel, CP-SAT feasibility
certificates with unsatisfiable cores, CUDA deterministic-template searches for
parity and modulo $3$, a nondeterministic path-relation GPU search closer to the
local-certificate definition, and group-run audits including a non-abelian $S_3$
case.  A final CUDA experiment extends the finite-certificate method to
context-free interval charts by measuring nonterminal-set and split-witness movement
for three small grammars.
\end{abstract}

\section{Introduction}
\label{sec:intro}

Regular languages are among the few computational objects for which several resource
notions admit algebraic classifications.  Static recognition factors through the
syntactic monoid, while strict and local testability describe membership by bounded
windows of the word.  This paper studies a different finite-word resource: an exact
local annotation that explains a regular-language membership value and remains stable
under one-symbol substitutions.  The annotation is auxiliary information, but it is not
trusted state.  After the edit, its correctness must still be forced by bounded-arity
local constraints over the input and the annotation.

The checker model itself is not new.  Stripped of the word-specific notation,
Definition~\ref{def:scheme} is a one-dimensional locally checkable proof: a prover
assigns finite labels, a local verifier checks bounded scopes and a root predicate, and
membership means that some accepting labelling exists.  Korman, Kutten, and Peleg
introduced proof-labeling schemes for distributed verification \cite{KKP10}.  Göös and
Suomela developed the locally checkable proof viewpoint and its proof-size hierarchy
\cite{GS16}.  Feuilloley and Fraigniaud study error-sensitive proof-labeling schemes,
where the number of detecting vertices is related to the distance from legality
\cite{FF17}.  Our contribution is not this certification framework.  It is the
edit-stability parameter on it: the Lipschitz constant of a chosen canonical accepting
labelling as the input word moves along one substitution edge.

Error sensitivity and edit stability measure different maps.  Error sensitivity fixes
one input instance and asks how many local checks reject an annotation far from every
legal annotation.  Edit stability compares chosen legal representatives of two
different input instances.  Neither definition gives a reduction to the other without
an additional theorem relating distance-to-legality inside one fibre to movement
between fibres.  This paper proves no such irreducibility theorem; establishing or
refuting one on paths is part of the model-comparison problem.

The algebraic source for bounded-window phenomena is older.  Brzozowski and Simon
characterised locally testable events \cite{BrzozowskiSimon73}.  Zalcstein gave an
independent early treatment of locally testable languages \cite{Zalcstein72}.  These
theories do not, however, answer the following question: when a membership answer is
accompanied by an externally checkable local explanation, how much of that explanation
is forced to change after a single edit?

This paper formalises that question for regular word languages.  The checker is exact,
deterministic, and local.  It receives the input word and a finite annotation; it accepts
only when all bounded-arity local consistency rules and one bounded root rule accept.
The annotation may be produced by any algorithm, but its correctness must be locally
forced from the input and the annotation itself.  After one substitution, the resource of
interest is the Hamming distance between canonical annotations, together with the number
of local constraints whose truth can change.  This differs from
trusted dynamic update time: a dynamic data structure may store a global bit that is
cheap to update but impossible for an external bounded-arity checker to force.

The distinction is already visible for parity.  A trusted structure toggles one bit per
edit.  A left-to-right prefix annotation is locally checkable but can change linearly
many cells.  A balanced product annotation changes only an ancestor path.  The point is
not that parity is difficult as a language; it is that association geometry exposes a
new stability parameter.  The relevant object is a locally checkable representation of a
monoid product, and the edit cost is the amount of that representation whose value is
forced to move.

As a data structure, the balanced product construction is the standard segment tree or
parallel-prefix association of an associative product.  The novelty claimed here is not
that construction itself, nor its $d$-ary height tradeoff, but its role as a locally
checkable certificate whose canonical representative has a measurable Lipschitz
response to input edits.

This paper is primarily a structural paper, but the finite search component also
needs reproducible certificates.  The ancillary material is organised as a small library rather
than as a single data dump.  Lean files expose the checked base case for the parity
product argument.  CUDA files expose both the product-atlas scans and a separate
deterministic path-template search, which is closer to the finite feasibility
formulation in Section~\ref{sec:finite}.  The certificates store aggregate obstruction
profiles and digests, because those profiles are the mathematical objects one wants to
inspect; raw traces would be much larger and less useful.

\paragraph{Main contributions.}
The paper makes the following contributions.
\begin{enumerate}[leftmargin=1.7em,itemsep=3pt]
  \item \textbf{An edit-stability parameter for local certification.}  Starting from
  the standard locally checkable proof/proof-labeling viewpoint, we define canonical
  annotations, cell and bit edit stability, and revalidation locality for regular
  words.  The new parameter can be read as a low-Lipschitz section problem over the
  substitution graph: the section chooses one locally accepted annotation for each
  word while keeping adjacent words close.  The bit metric records information hidden
  by enlarging the label alphabet.

  \item \textbf{Universal logarithmic stability.}  Every regular language recognised by
  a finite morphism $\eta:\Sigma^{*}\to M$ admits a balanced product annotation of
  constant locality, linear size, $\bigO(\log n)$ edit stability, $\bigO(\log n)$
  incremental revalidation, and constant-time access to the membership value.

  \item \textbf{Algebraic lower bounds inside projected product geometries.}  The
  construction extends to bounded fan-in $d$, giving the expected arity-height tradeoff.
  Conversely, if a fan-in-$d$ product decomposition exposes an edit-active nontrivial
  group quotient at its leaves and internal product labels, then one substitution
  forces $\Omega(\log_d n)$ changed quotient labels.  This generalises the parity
  product-tree lower bound, without claiming a lower bound for arbitrary
  auxiliary annotations.

  \item \textbf{The annotation-free boundary.}  Bounded-window annotation-free
  recognition is exactly strict locality.  Hence strictly local languages have zero
  annotation size and constant revalidation, whereas even a one-bit group language has
  no annotation-free bounded-window recogniser.

  \item \textbf{Closure bookkeeping and aligned edits.}  We separate the original
  one-sided existential model from a two-sided total decision variant.  In the
  two-sided variant we prove closure under Boolean operations, fixed word quotients,
  finite products of decision schemes, and uniform inverse morphisms in the
  substitution model.  In the aligned edit-trace model the same proof handles
  arbitrary fixed morphisms.  We also give an aligned insertion/deletion metric and
  show that balanced product annotations retain logarithmic movement in that model.

  \item \textbf{A finite obstruction method.}  For fixed arity, label alphabet,
  stability budget, scopes, and maximum length, the existence of a stable exact local
  annotation family is a finite Boolean/integer feasibility problem.  This gives a
  precise way to falsify proposed constant-stability templates at finite scale and to
  extract obstruction patterns before proving an asymptotic classification.

  \item \textbf{A structured ancillary library.}  The ancillary files include a Lean~4
  development, checked without external libraries, formalising the parity product
  kernel used in the group lower-bound argument.  They also include CUDA pipelines
  and certificate schemas for finite product-annotation atlases, deterministic
  path-local template searches, and nondeterministic path-relation searches that
  materialise 32\,GiB tables on an A100, plus a context-free interval-chart pipeline
  that materialises nonterminal and split-witness movement.
\end{enumerate}

The paper is explicit about the boundary it does not settle.  The logarithmic upper
bound is a theorem for all regular languages.  The lower bound is a theorem for
quotient-exposing product decompositions, not for every possible local annotation
scheme.  The full classification of regular languages with constant edit-stable exact
local annotations remains open.  The contribution is to isolate this boundary, supply
the universal upper bound, identify an algebraic source of logarithmic lower bounds
inside product geometries, and turn the remaining boundary into a finite obstruction
problem.

Trusted dynamic maintenance of regular languages is a related but different resource
model: it permits auxiliary state maintained by update formulas, whereas the present
model requires the auxiliary information to be locally checkable after each edit.  The
relevant dynamic-maintenance classifications are used only for comparison
\cite{PatnaikImmerman94,GMS12,Tschirbs23,BTVZ26}.

\paragraph{Relation to GPU automata work.}
There is a substantial systems literature on running automata and regular-expression
engines on GPUs.  Cascarano, Rolando, Risso, and Sisto introduced iNFAnt as an NFA
pattern-matching engine for GPGPU devices \cite{Cascarano10}.  Yu and Becchi compared
automata representations for regular-expression matching on GPUs \cite{YuBecchi13}.
Mytkowicz, Musuvathi, and Schulte studied data-parallel finite-state machines as a
general execution model rather than a single regex application \cite{Mytkowicz14}.  Liu,
Pai, and Jog analysed why NFAs execute poorly on GPUs and how scheduling choices can
make them faster \cite{LiuPaiJog20}.  Ge, Zhang, and Liu later proposed ngAP for
large-scale non-blocking automata processing on GPUs \cite{GeZhangLiu24}.  Valizadeh
and Berger used GPUs for search-based regular-expression inference \cite{ValizadehBerger23}.
HybridSA combines GPU execution with bit parallelism for multi-pattern regex matching
\cite{HybridSA24}.  Those papers study throughput, memory layout, scheduling, and
synthesis performance.  The present use of the GPU is different: the automata are
small, but the finite obstruction spaces are large and naturally data-parallel.  One
ancillary pipeline materialises canonical product annotations and measures their
Lipschitz profiles.  A second pipeline enumerates deterministic path-local templates,
materialises their length-$n$ final-state table in VRAM, and filters the templates that
are exact for parity up to that length.  The theoretical claims in the paper do not
depend on these runs.  The runs supply reproducible finite objects for testing proposed
lower-bound mechanisms and for refining the open boundary problem.

\section{Preliminaries}
\label{sec:prelim}

\paragraph{Words and substitutions.}
The alphabet $\Sigma$ is finite.  The free monoid is $\Sigma^{*}$, with empty word
$\eps$.  For $w=w_1\cdots w_n\in\Sigma^n$, $i\in[n]$ and $a\in\Sigma$, write
$\mathrm{sub}_{i}^{a}(w)$ for the word obtained by replacing $w_i$ by $a$.  The length
is fixed under this edit operation, so annotation cells can be compared before and
after the edit.

\paragraph{Recognising morphisms.}
A recogniser for $L$ is a triple $(M,\eta,F)$ where $M$ is a finite monoid,
$\eta:\Sigma^{*}\to M$ is a morphism, $F\subseteq M$, and
$L=\eta^{-1}(F)$.  The least recogniser is the syntactic morphism onto the syntactic
monoid $\synt{L}$, obtained from the congruence
\[
  u\equiv_L v
  \quad\Longleftrightarrow\quad
  \forall x,y\in\Sigma^{*},\; xuy\in L \Longleftrightarrow xvy\in L .
\]
The constructions below work for any recogniser; using $\synt{L}$ minimises the monoid
constant.

\begin{lemma}[Quotients commute with product annotations]
\label{lem:quotientannotation}
Let $\rho:M\to N$ be a monoid quotient and let $\eta_N=\rho\eta_M$.  The balanced
product annotation over $M$ projects cellwise under $\rho$ to the balanced product
annotation over $N$.
\end{lemma}

\begin{proof}
At a leaf, $\rho(\eta_M(w_i))=\eta_N(w_i)$.  At an internal node, $\rho(xy)=\rho(x)\rho(y)$,
so induction over the product tree gives the claim.
\end{proof}

Thus the choice of recogniser affects constants in the label alphabet, but not the
asymptotic locality statements proved below.

This lemma is intentionally only a statement about the product annotation.  It is not
a closure theorem for the class of languages admitting arbitrary constant-stability
schemes.  Definition~\ref{def:scheme} permits several inequivalent accepting
annotations for the same word, and a general scheme need not carry a cellwise quotient
map.  Consequently the closure properties needed for a variety theorem are separate
open problems, not consequences of Lemma~\ref{lem:quotientannotation}.

\paragraph{Local and algebraic language classes.}
A language is \emph{strictly $k$--local} if membership of sufficiently long words is
specified by a set of allowed length--$k$ factors together with allowed
length--$(k-1)$ prefixes and suffixes.  The union over $k$ is denoted $\St$.  A language
is \emph{locally testable}, denoted $\Lt$, if membership depends on the set of length
$k$ factors and the two boundary words of length $k-1$, for some $k$.  The classical
algebraic theory gives
\[
  \St\subsetneq\Lt\subsetneq\Ap,
\]
where $\Ap$ is the class of aperiodic syntactic monoids; star--free languages are
exactly the aperiodic regular languages.  Sch\"utzenberger proved the algebraic
aperiodicity theorem for star-free languages \cite{Schutzenberger65}.  McNaughton and
Papert gave the logical counter-free automata account \cite{McNaughtonPapert71}.
Brzozowski and Simon developed one of the classical characterisations of locally
testable events \cite{BrzozowskiSimon73}.  Zalcstein gave an independent early account
of locally testable languages \cite{Zalcstein72}.  Eilenberg's variety theory supplies
the general language-to-monoid framework \cite{Eilenberg76}, and Pin's monograph is the
standard reference for these pseudovariety classifications \cite{Pin86}.  The class
$\ZG$ is the monoid class appearing in constant--time trusted dynamic membership for
regular languages \cite{AJP21}.

\paragraph{Balanced product trees.}
A balanced ordered binary tree over $[n]$ has leaves $1,\dots,n$ from left to right and
height at most $\lceil\log_2 n\rceil+1$.  One obtains such a tree by recursively
splitting each interval into two nearly equal subintervals.  Each internal node
corresponds to a contiguous interval, its children partition that interval into a left
and a right interval, and a full binary tree with $n$ leaves has exactly $2n-1$ nodes.

\section{Local annotation schemes}
\label{sec:schemes}

The model separates three objects: the input word, an annotation, and a checker.  The
annotation may be produced by any method; the checker sees only the input symbols and
the annotation cells named by its local rules.

\begin{definition}[Local annotation scheme]
\label{def:scheme}
Fix a recogniser $(M,\eta,F)$ for $L$.  A local annotation scheme assigns to every
length $n$ the following data.
\begin{enumerate}[label=(\alph*),leftmargin=2em,itemsep=2pt]
  \item A finite set $I(n)$ of annotation cells and a fixed finite label alphabet
  $\Lambda$, independent of $n$.  An annotation is a map
  $\gamma:I(n)\to\Lambda$.

  \item A finite set $\Pi(n)$ of local rules.  Each rule reads at most $r$ annotation
  cells and at most $r$ input symbols, and evaluates a fixed predicate on the values it
  reads.  The arity bound $r$ is independent of $n$.

  \item A root rule reading at most $r$ annotation cells.
\end{enumerate}
The checker accepts $(w,\gamma)$ iff all local rules and the root rule accept.  The
scheme recognises $L$ with perfect soundness when, for every $w\in\Sigma^n$,
\[
  \exists \gamma\; C(w,\gamma)=1
  \quad\Longleftrightarrow\quad
  w\in L .
\]
\end{definition}

The definition does not impose a running time on the process that finds $\gamma$.  The
point is external checkability: once an annotation is proposed, correctness follows
from bounded--arity constraints.

\begin{definition}[Canonical annotation and costs]
\label{def:costs}
A canonical annotation map for a scheme is a function
$w\mapsto\gamma_w$ defined for all words of length $n$, such that $\gamma_w$ is
accepted whenever $w\in L$.  Relative to such a map define
\begin{align*}
  \size(n) &:= |I(n)|\cdot \lceil \log_2 |\Lambda|\rceil,\\
  \loc &:= r,\\
  \edit(n) &:=
  \max_{w\in\Sigma^n}\max_{i\in[n],a\in\Sigma}
  \left|\{j\in I(n):\gamma_w(j)\ne\gamma_{\mathrm{sub}_{i}^{a}(w)}(j)\}\right|.
\end{align*}
We also write $\rev(n)$ for the maximum number of local rules whose truth value has to
be rechecked after a substitution and the corresponding canonical annotation update.
\end{definition}

The stability measure counts changed cells, not the bit distance inside a cell.  This
is the natural analogue of node changes in a dynamic labelled data structure.  The
fixed-alphabet condition prevents one cell from encoding the entire length-$n$ word.
An extended variable-alphabet model is possible, but then cell movement must be read
together with bit movement and annotation size.

\begin{definition}[Bit edit stability]
\label{def:bitstability}
Fix an injective encoding
$c:\Lambda\to\{0,1\}^{b}$, where
$b=\lceil\log_2|\Lambda|\rceil$.  The bit edit stability of a canonical section is
\[
  \bedit(n)=
  \max_{ww'\in E(Q_n)}
  \sum_{j\in I(n)}
  \dist_{\Ham}\!\left(c(\gamma_w(j)),c(\gamma_{w'}(j))\right).
\]
\end{definition}

\begin{lemma}[Cell/bit comparison]
\label{lem:cellbit}
For every injective fixed-length label encoding,
\[
  \edit(n)\le \bedit(n)\le
  \edit(n)\,\lceil\log_2|\Lambda|\rceil.
\]
\end{lemma}

\begin{proof}
Every changed cell changes at least one encoded bit and at most all
$\lceil\log_2|\Lambda|\rceil$ encoded bits.  Summing over cells and taking the
maximum over substitution edges gives both inequalities.
\end{proof}

Thus a small cell movement is meaningful only together with the label size.  In the
main fixed-alphabet model, cell and bit movement differ by only a constant factor.  In
the extended variable-alphabet model, enlarging the alphabet can pack several algebraic
values into one changed cell, but the packed information remains visible in
$\size(n)$ and $\bedit(n)$.  The natural invariant is
therefore the resource profile
\[
  \bigl(\size(n),\loc,\edit(n),\bedit(n),\rev(n)\bigr),
\]
not $\edit(n)$ in isolation.

\paragraph{Rule locality and non-encoding.}
The rule family may depend on the length $n$, but each rule has bounded arity.  The
content of a rule is a finite predicate on the fixed-alphabet values and input symbols
in its scope.  It cannot read an unbounded address, an unbounded counter, or a
length-dependent label carrying the entire word.  A large number of local rules is
permitted, just as a circuit may have many gates, but no single rule performs a global
computation.

\paragraph{Canonical annotations for nonmembers.}
The canonical map is defined on all words, not only on members.  For the product
annotations below, nonmembers receive the same locally consistent monoid labels as
members; the only failing condition is the root membership test.  This convention is
important for edit stability, because a substitution may move a word into or out of the
language while the cell identities remain fixed.

\paragraph{A Lipschitz-section viewpoint.}
For each length $n$, let $Q_n$ be the substitution graph on $\Sigma^n$: two words are
adjacent when they differ in one position.  A canonical annotation map is a section
\[
  \Gamma_n:\Sigma^n\longrightarrow \Lambda^{I(n)},
  \qquad \Gamma_n(w)=\gamma_w,
\]
of the local acceptance relation chosen by the checker.  The stability cost is exactly
the Lipschitz constant of this section from $Q_n$ with graph distance one to the
Hamming metric on annotation cells:
\[
  \edit(n)=\max_{ww'\in E(Q_n)}
  \dist_{\Ham}(\Gamma_n(w),\Gamma_n(w')).
\]
Thus the problem is not whether membership can be updated, but whether the global
regular invariant admits a locally checkable realisation whose accepted representatives
move little on every substitution edge.

\paragraph{Four costs, four different questions.}
The quantities $\size$, $\loc$, $\edit$, and $\rev$ are intentionally separated.
Annotation size measures stored algebraic information; locality measures the largest
arity of a local consistency test; edit stability measures how many canonical labels
change; revalidation locality measures how many local constraints have to be checked
again after the canonical update.  A structure can be excellent for one of these costs
and poor for another.  The prefix annotation, for example, has the same label alphabet
as the balanced tree but a much worse edit profile.

\paragraph{Worst-case versus distributional movement.}
This paper uses $\edit(n)$ for the worst-case Lipschitz constant because it is the
right notion for adversarial local certification: every substitution edge must have a
nearby accepted representative.  It is not the only possible stability invariant.  For
a distribution $\mu_n$ on substitution edges one can also study
\[
  \mathbb E_{ww'\sim\mu_n}
  \dist_{\Ham}(\Gamma_n(w),\Gamma_n(w')),
\]
and for edit sequences one can ask for amortised movement.  These distributional
variants can separate behaviours that worst-case movement collapses.  In the ancillary
product atlas, for example, group rows have mean movement close to their maximum,
whereas locally testable rows have much smaller mean movement despite the same product
worst case.  The present lower-bound theorems are worst-case statements, while
Proposition~\ref{prop:meanmovement} records the product-geometry theorem currently
available for the expected version.  Extending that distributional theory beyond
product annotations is left open rather than implied by the worst-case results.

\begin{remark}[Exact deterministic checking necessarily touches the input]
\label{rem:readall}
A deterministic exact checker that never inspects a position on which membership can
depend cannot have perfect soundness: the uninspected symbol can be changed while the
annotation is kept fixed.  Sublinear verification for language membership therefore
requires a weaker guarantee, usually proximity soundness and randomness.  The model
studied here keeps soundness exact and deterministic, and pays for it by allowing a
linear number of local rules, each of bounded arity.
\end{remark}

\section{Two universal annotations}
\label{sec:universal}

Let $(M,\eta,F)$ recognise $L$.  Both schemes in this section use monoid elements as
labels.  The first is the usual left--to--right product.  The second is a balanced
association of the same product.

\subsection{The prefix annotation}

\begin{definition}[Prefix annotation]
\label{def:prefix}
Let $I(n)=\{0,1,\dots,n\}$ and $\Lambda=M$.  The canonical annotation is
\[
  \gamma_w(i)=\eta(w_1\cdots w_i),\qquad 0\le i\le n,
\]
with $\gamma_w(0)=1_M$.  The local rules are
\[
  \gamma(0)=1_M,
  \qquad
  \gamma(i)=\gamma(i-1)\eta(w_i)\quad (1\le i\le n),
\]
and the root rule is $\gamma(n)\in F$.
\end{definition}

\begin{lemma}[Correctness and instability of the prefix annotation]
\label{lem:prefix}
The prefix annotation recognises $L$ with perfect soundness, has locality at most $3$,
and has size $(n+1)\lceil\log_2|M|\rceil$.  For every recognised language its stability
is at most $n+1$.  This upper bound is sharp up to constants: for parity over
$\{a,b\}$, with $\eta(a)=1$ and $\eta(b)=0$ in $\mathbb Z/2\mathbb Z$, the stability is
$\Theta(n)$.
\end{lemma}

\begin{proof}
If all local rules hold, induction gives
$\gamma(i)=\eta(w_1\cdots w_i)$ for each $i$.  Hence the root rule holds exactly when
$\eta(w)\in F$, i.e. when $w\in L$.  This proves both completeness and perfect
soundness.

A substitution at position $i$ cannot change labels $0,\dots,i-1$ and can only change
labels $i,\dots,n$, so the stability is at most $n+1$.  For parity take $w=b^n$ and
edit the first symbol to $a$.  All nonzero prefix labels switch from $0$ to $1$, so
$n$ cells change.
\end{proof}

\subsection{The balanced product annotation}

\begin{definition}[Balanced product annotation]
\label{def:tree}
Fix a balanced ordered binary tree $T_n$ over the input positions.  The cell set
$I(n)$ is the node set of $T_n$ and $\Lambda=M$.  For a leaf $i$ set
$\gamma_w(i)=\eta(w_i)$.  For an internal node $v$ with left child $\ell$ and right
child $r$, set
\[
  \gamma_w(v)=\gamma_w(\ell)\gamma_w(r).
\]
Equivalently, $\gamma_w(v)$ is the ordered product of the letters below $v$.  Local
rules require each leaf label to equal $\eta(w_i)$ and each internal label to equal the
product of its two children.  The root rule is $\gamma(\rootnode)\in F$.
\end{definition}

\begin{theorem}[Universal logarithmic annotation]
\label{thm:universal}
Every regular language recognised by $(M,\eta,F)$ admits the balanced product
annotation with perfect soundness and
\[
  \loc\le 3,
  \qquad
  \size(n)=(2n-1)\lceil\log_2|M|\rceil,
\]
\[
  \edit(n)\le \lceil\log_2 n\rceil+\bigO(1),
  \qquad
  \rev(n)=\bigO(\log n).
\]
The updated root label gives the membership answer in constant time.
\end{theorem}

\begin{proof}
Suppose the checker accepts $(w,\gamma)$.  The leaf rules force
$\gamma(i)=\eta(w_i)$ at every leaf.  By induction on node height, every internal rule
then forces $\gamma(v)$ to equal the ordered monoid product of the input letters in the
interval below $v$.  At the root this product is $\eta(w)$.  The root rule therefore
holds exactly when $w\in L$.  Thus the scheme has perfect soundness and completeness.

Every local rule reads at most three labels or one label and one input symbol.  The
number of nodes is $2n-1$ up to the harmless identity padding convention, giving the
size bound.

After changing position $i$, only the leaf $i$ and its ancestors can have a different
canonical value.  Some ancestors may remain numerically equal in a noncancellative
monoid, so the correct statement is an upper bound rather than equality.  The number of
possible changed cells is at most the height plus one, namely $\lceil\log_2 n\rceil+
\bigO(1)$.  Recomputing these labels bottom--up costs one monoid multiplication per
internal ancestor.  The only local rules whose truth value can change are the leaf rule
at $i$ and the internal product rules incident with this ancestor path, so the
incremental revalidation cost is also logarithmic.
\end{proof}

\subsection{Distributional product movement}
\label{subsec:distributional}

Worst-case stability is used for adversarial edits.  The same
canonical product annotation also has a distributional profile.  For a fixed canonical
map $\Gamma_n$, let
\[
  \avg(n)=
  \mathbb E_{w\in\{0,1\}^n,\;i\in[n]}
  \dist_{\Ham}\!\left(\Gamma_n(w),\Gamma_n(w\oplus e_i)\right),
\]
where the word and edited position are uniform and $w\oplus e_i$ flips the $i$th bit.

\begin{proposition}[Mean product movement separates cancellation from absorption]
\label{prop:meanmovement}
For the balanced product annotation over $\{0,1\}$, the parity group
$\mathbb Z/2\mathbb Z$ has $\avg(n)=\Theta(\log n)$.  For the aperiodic semilattice
language ``contains $1$'', recognised by $(\{0,1\},\vee)$ with accepting value $1$, the
same balanced product geometry has $\avg(n)<2$ for every $n$.
\end{proposition}

\begin{proof}
For parity, flipping any bit changes the projected value of every product interval
containing that bit: in the group, left and right cancellation reduce equality of the
old and new interval products to equality of the two edited letter values.  Thus the
number of changed labels is exactly the number of nodes on the edited leaf-to-root
path.  In a balanced binary tree the average such path length is $\Theta(\log n)$.

For ``contains $1$'', a node label is the disjunction of the bits in its interval.  Fix
an edited position $i$ and an ancestor interval $J$ of size $m$.  The label at $J$
changes after flipping $w_i$ iff all other $m-1$ bits in $J$ are zero.  Under the
uniform word distribution this event has probability $2^{-(m-1)}$.  Along a root path
the interval sizes are strictly increasing, so for the fixed position
\[
  \mathbb E[\text{changed ancestor labels at }i]
  \le \sum_{m\ge1}2^{-(m-1)}=2.
  \]
The strict inequality follows because the finite root path omits all but its actual
interval sizes.  Averaging over $i$ preserves the bound.
\end{proof}

This proposition is not an optimality statement for aperiodic languages.  The language
``contains $1$'' is locally testable, not strict-local under the bounded-window
definition used below, and its worst-case product movement is still the tree height.
The proposition has the following limited role: inside the same product geometry,
group cancellation gives logarithmic mean movement, whereas idempotent absorption can
keep the mean bounded.

\subsection{Bounded fan--in product geometries}
\label{subsec:dary}

The binary tree is not essential.  It is the first point on a tradeoff between local
arity and update height.

\begin{definition}[$d$--ary product annotation]
\label{def:dary}
Fix an integer $d\ge 2$.  A $d$--ary product tree over $[n]$ is an ordered rooted tree
whose leaves are $1,\dots,n$, whose internal nodes have between two and $d$ children,
and whose height is $\lceil\log_d n\rceil+\bigO(1)$.  A node is labelled by the ordered
monoid product of the leaves in its interval.  The local rule at an internal node says
that its label is the product of the labels of its children in left--to--right order;
the leaf and root rules are as in Definition~\ref{def:tree}.
\end{definition}

\begin{proposition}[Arity--height tradeoff]
\label{prop:dary}
For every fixed $d\ge2$, the $d$--ary product annotation recognises $L$ with perfect
soundness and satisfies
\[
  \loc\le d+1,
  \qquad
  \size(n)=\bigO_d(n\log |M|),
  \qquad
  \edit(n)=\bigO(\log_d n),
  \qquad
  \rev(n)=\bigO(d\log_d n).
\]
The root label still answers membership in constant time.
\end{proposition}

\begin{proof}
The correctness proof is the same induction as in Theorem~\ref{thm:universal}, with
one multiplication constraint of arity at most $d+1$ at each internal node.  A single
substitution can change only the edited leaf and its ancestors.  The number of such
nodes is bounded by the tree height.  Revalidating a changed internal node may require
reading all of its children, hence the factor $d$ in the displayed bound.  For fixed
$d$ this is logarithmic in $n$.
\end{proof}

\begin{definition}[Edit-active group quotient]
\label{def:activegroup}
Let $(M,\eta,F)$ recognise $L$.  A finite group quotient $\pi:M\to G$ is
\emph{edit-active} for the substitution model if there exist letters $a,b\in\Sigma$ with
$\pi\eta(a)\ne\pi\eta(b)$.  It is
\emph{membership-visible} if some such pair $a,b$ is separated by the language: there
are contexts $x,y\in\Sigma^{*}$ such that $xay\in L$ and $xby\notin L$, or conversely.
\end{definition}

The activity condition says that a single permitted edit can change the group component.
The visibility condition is not needed for the product-label lower bound, but it is what
makes the quotient relevant to the language rather than to an arbitrary nonminimal
recogniser.  For nonminimal recognisers it is an extra hypothesis, not automatic.

\begin{lemma}[Syntactic activity is membership-visible]
\label{lem:syntacticvisible}
Let $\eta_L:\Sigma^*\to\synt{L}$ be the syntactic morphism.  If a quotient
$\pi:\synt{L}\to G$ is edit-active, then it is membership-visible.
\end{lemma}

\begin{proof}
Choose letters $a,b$ with $\pi\eta_L(a)\ne\pi\eta_L(b)$.  Then
$\eta_L(a)\ne\eta_L(b)$ in the syntactic monoid.  By the definition of the syntactic
congruence, two distinct syntactic classes are separated by some contexts: there are
$x,y\in\Sigma^*$ such that $xay$ and $xby$ have different membership values in $L$.
\end{proof}

\begin{proposition}[Projected group lower bound in product decompositions]
\label{prop:grouplower}
Fix $d\ge2$ and suppose that $(M,\eta,F)$ has an edit-active quotient onto a nontrivial
finite group $G$.  Consider any $G$-projected ordered product decomposition whose
dependency graph is a rooted tree of maximum fan-in $d$, whose leaf labels are the
projected elements $\pi\eta(w_i)$, and whose internal labels are the ordered products of
their children in $G$.  Then for every $n$ there is a one-symbol substitution after
which at least $\lceil\log_d n\rceil+1$ projected product labels change.
\end{proposition}

\begin{proof}
Every rooted tree with $n$ leaves and fan-in at most $d$ has a leaf $i$ of depth at
least $\lceil\log_d n\rceil$.  Choose letters $a,b$ with $\pi\eta(a)\ne\pi\eta(b)$ and
put $a$ at every leaf.  Change the deepest leaf $i$ from $a$ to $b$.  Consider any
ancestor interval of $i$.  Its old projected product has the form $u\,\pi\eta(a)\,v$ in
$G$, and its new product has the form $u\,\pi\eta(b)\,v$, where $u$ and $v$ are the
products of the leaves in the same interval to the left and right of $i$.  If these two
products were equal, multiplying by $u^{-1}$ on the left and by $v^{-1}$ on the right
would give $\pi\eta(a)=\pi\eta(b)$, a contradiction.  Hence the label at every ancestor
of $i$, together with the leaf label itself, changes.  The number of changed labels is
therefore at least the number of nodes on this root path, which is
$\lceil\log_d n\rceil+1$.
\end{proof}

\begin{corollary}[Parity as the minimal group obstruction]
\label{prop:productlower}
For parity over $\{a,b\}$, recognised by $\mathbb Z/2\mathbb Z$ with distinct projected
letter values, any bounded-degree product decomposition that exposes the quotient
product labels has worst-case edit stability $\Omega(\log n)$.  The balanced product
annotation attains this bound up to constants.
\end{corollary}

This lower bound does not rule out non-product annotations, nor product annotations
that hide the quotient inside additional side information.  It says that the usual
associative-product geometry cannot asymptotically improve on balanced height when the
scheme commits to carrying the edited group quotient as projected product labels.  That
distinction is essential: the broader constant-stability question is a question about
all bounded-arity annotation schemes, not only about segment trees.

\begin{remark}[Relation to Merkle association]
\label{rem:merkle}
The construction has the shape of a Merkle tree, but the label operation is exact
monoid multiplication rather than a hash.  Consequently the soundness argument is
information--theoretic: an accepting annotation is forced to contain the true monoid
value of every interval.  As a computation pattern it is the standard parallel--prefix
association of an associative product, classically studied by Ladner and Fischer
\cite{LadnerFischer80}.  Its tree shape resembles Merkle's authenticated hash tree,
but here the labels are exact algebraic products rather than cryptographic hashes
\cite{Merkle87}.
\end{remark}

\begin{theorem}[Prefix versus balanced association]
\label{thm:separation}
For parity, the prefix annotation has stability $\Theta(n)$, while the balanced product
annotation has stability $\bigO(\log n)$.  Hence edit stability is a property of the
chosen annotation geometry, not merely of the recognised language or the recognising
monoid.
\end{theorem}

\begin{proof}
The lower bound for the prefix annotation is Lemma~\ref{lem:prefix}.  The logarithmic
upper bound for the balanced annotation is Theorem~\ref{thm:universal} applied to
$M=\mathbb Z/2\mathbb Z$.
\end{proof}

\section{The annotation--free boundary}
\label{sec:frontier}

A scheme is annotation--free when $|\Lambda|=1$.  Such a scheme has no
informative cells; all work is done by the input windows read by the local rules.  To
avoid encoding nonlocal information in the choice of the rule family, this section uses
the standard bounded--window convention: rules are translates of finitely many
predicates on contiguous windows, with separate finitely many boundary predicates.

\begin{definition}[Bounded--window recognition]
\label{def:window}
A bounded--window recogniser of width $k$ consists of a set $P\subseteq\Sigma^{k}$ of
allowed factors, a set $B_{\ell}\subseteq\Sigma^{k-1}$ of allowed prefixes, a set
$B_r\subseteq\Sigma^{k-1}$ of allowed suffixes, and finitely many exceptional decisions
for words of length $<k$.  A word $w$ of length at least $k-1$ is accepted iff its
length--$(k-1)$ prefix lies in $B_\ell$, its length--$(k-1)$ suffix lies in $B_r$, and
every length--$k$ factor of $w$ lies in $P$.
\end{definition}

\begin{proposition}[Strict locality is exactly annotation--free bounded--window
recognition]
\label{prop:slt}
A regular language has an annotation--free bounded--window recogniser if and only if it
is strictly locally testable.  In that case $\size(n)=0$, $\edit(n)=0$, and a single
substitution affects only $\bigO(k)$ window rules.
\end{proposition}

\begin{proof}
Definition~\ref{def:window} is precisely the usual definition of strict local
testability, with the finite set of short words handled separately.  The associated
checker has no informative annotation cells, so the size and stability are zero.  A
substitution at position $i$ can affect only windows starting in
$\{i-k+1,\dots,i\}$, plus the two boundary tests if $i$ is near an endpoint.  Thus the
incremental revalidation cost is bounded by a constant depending only on $k$.
\end{proof}

The next proposition records the cleanest separation between trusted maintenance and
annotation--free local recognition.

\begin{proposition}[Trusted constant update does not imply annotation--free locality]
\label{prop:gap}
There exists a regular language with constant trusted dynamic update time under
substitutions but with no annotation--free bounded--window recogniser.
\end{proposition}

\begin{proof}
Let $L=\{w\in\{a,b\}^{*}: |w|_a\text{ is even}\}$.  A trusted dynamic structure stores
$|w|_a\bmod 2$ and toggles it exactly when the substitution changes whether the edited
symbol is $a$.  This is constant time.  Algebraically, the syntactic monoid is
$\mathbb Z/2\mathbb Z$, a nontrivial group.  Therefore $L$ is not aperiodic and hence
not locally testable.  By Proposition~\ref{prop:slt}, it has no annotation--free
bounded--window recogniser.
\end{proof}

\begin{problem}[The constant--stability boundary and its closure properties]
\label{prob:boundary}
Classify the regular languages admitting a perfectly sound bounded--arity annotation
scheme with $\edit(n)=\bigO(1)$ and $\rev(n)=\bigO(1)$.  Is this class exactly
$\St$, does it contain any locally testable language not in $\St$, and is it closed
under the operations needed for an algebraic variety theorem?
\end{problem}

The last clause is essential.  The set of languages admitting a constant-stability
scheme is not known to be closed under inverse morphisms, quotients, Boolean
operations, or finite products.  Lemma~\ref{lem:quotientannotation} proves only that
the canonical product annotation commutes with monoid quotients.  It does not prove
closure of the abstract resource class, because a general scheme may use non-product
cells and may have many accepted annotations for the same word.  Thus it would be
premature to ask whether the class is a variety strictly between $\Lt$ and $\Ap$
before these closure properties are established.

There is, however, a clean closure statement for the two-sided version of the model.
It is useful because it identifies exactly where the one-sided existential definition
loses the algebraic bookkeeping required by Eilenberg-style classifications.

\begin{definition}[Two-sided total local decision scheme]
\label{def:twosided}
A two-sided total scheme for $L$ is a bounded--arity local checker
$C(w,\gamma,b)$ with a claimed output bit $b\in\{0,1\}$ such that, for every word
$w$ and bit $b$,
\[
  \exists\gamma\; C(w,\gamma,b)=1
  \quad\Longleftrightarrow\quad
  b=\mathbf 1_L(w).
\]
A canonical section chooses an accepted annotation for the unique correct bit of
each word, and its costs are measured as in Definition~\ref{def:costs}.  The
claimed bit is a constant-size root value and does not affect asymptotic cell
movement.
\end{definition}

The balanced product annotation is two-sided total: the local product constraints
force the root monoid value, and the root rule checks whether the claimed bit agrees
with membership in $F$.

\begin{proposition}[Closure for two-sided total schemes]
\label{prop:twosidedclosure}
The class of languages with two-sided total schemes of bounded locality and
constant edit stability is closed under Boolean operations and under left and right
quotients by fixed words.  It is also closed under inverse images of fixed uniform
morphisms in the fixed-length substitution model.  If the schemes are measured in the
aligned edit-trace model, then closure holds under inverse images of arbitrary fixed
morphisms.  The same statements hold with logarithmic stability in place of constant
stability.
\end{proposition}

\begin{proof}
For a Boolean combination $f(L_1,\dots,L_m)$ over the same input alphabet, take the
disjoint union of the annotation cells for the $m$ schemes.  The combined certificate
contains claimed bits $b_1,\dots,b_m$, locally checks each component certificate, and
has a root rule requiring $b=f(b_1,\dots,b_m)$.  Exactness follows because each
component has a certificate for exactly its true bit.  A canonical section is the
product of the component sections, so edit movement and revalidation add over the
fixed number of components.  Thus constant bounds remain constant and logarithmic
bounds remain logarithmic.

For a fixed two-sided quotient $x^{-1}Ly^{-1}=\{w:xwy\in L\}$, simulate the scheme for
$L$ on the word $xwy$.  The added context has constant length, and a substitution in
$w$ changes one position of the simulated word.  The annotation size changes by only a
constant boundary term, while stability and revalidation are bounded by the
corresponding costs for $L$ at length $|w|+|x|+|y|$.

Finally let $h:\Delta^*\to\Sigma^*$ be $k$-uniform.  On input
$u\in\Delta^n$, simulate the scheme for $L$ on the expanded word $h(u)$ of length
$kn$.  Every local rule reading an output symbol of $h(u)$ is implemented by reading
the corresponding input letter and the fixed offset inside its $k$-letter block.  A
single substitution in $u$ changes at most $k$ output positions.  Moving between the
two expanded words by these $k$ substitutions and applying the triangle inequality
gives edit cost at most $k\,\edit_L(kn)$ and revalidation cost at most
$k\,\rev_L(kn)$, up to a constant factor depending on $k$ and the local arity.

For an arbitrary fixed morphism $h$, let
$K=\max_{a\in\Delta}|h(a)|$.  A substitution $a\mapsto b$ in the input replaces the
block $h(a)$ by the block $h(b)$ in the simulated word.  In the aligned edit-trace
model this replacement is a bounded sequence of at most $|h(a)|+|h(b)|\le 2K$
insertions, deletions, and substitutions, with the surrounding blocks aligned
identically.  Applying the target scheme's aligned-edit bound along this constant
length edit trace gives the claimed closure.  Erasing letters are covered by the same
argument, because deleting the empty block or inserting it contributes no simulated
positions.
\end{proof}

Proposition~\ref{prop:twosidedclosure} is intentionally not stated for arbitrary
one-sided schemes.  In the original model, a nonmember need not have any accepted
annotation, so complementing or combining languages may require certificates for
negative information that the scheme never provided.  The proposition also explains
the exact status of non-uniform inverse morphisms.  They are not an additional algebraic obstacle:
they are ordinary inverse morphisms once the metric is the aligned string-edit metric.
They fall outside the fixed-length substitution graph only because one input
substitution can become a bounded mixture of substitutions, insertions, and deletions
in the expanded word.

The universal tree gives the $\bigO(\log n)$ ceiling for all regular languages, while
Proposition~\ref{prop:slt} gives the zero--annotation part of the constant side.  The
intermediate locally testable languages are subtle.  A language such as ``some factor
from a fixed set occurs'' is locally testable, but exact local verification must prevent
both false absence and false presence of the relevant factor.  A single global flag is
not locally forced; a naive distributed mark for occurrence behaves like a path to a
marked position and can be unstable.  Thus the upper side beyond strict locality is not
a formality.  The obstruction may be algebraic, combinatorial, or a mixture of both.

The expected source of lower bounds outside the constant-stability class is the same kind of
long-distance dependence that appears in dynamic membership: when the syntactic monoid
contains nontrivial group behaviour, local updates must in some form transport prefix
information over long distances.  Proposition~\ref{prop:productlower} proves this only
inside product-decomposition schemes, so the general lower-bound problem remains open.

\section{Finite obstruction search}
\label{sec:finite}

This section turns Problem~\ref{prob:boundary} into a finite problem at each fixed
scale.  The point is not to replace an asymptotic proof, but to make small stable
schemes falsifiable by exhaustive finite checks over the Cayley table of the monoid.

Fix a finite alphabet $\Sigma$, a recogniser $(M,\eta,F)$, an arity bound $r$, a label
alphabet $\Lambda$, a stability budget $s$, and a maximum length $N$.  A \emph{template}
for lengths at most $N$ consists of cell sets $I(n)$, local rule scopes, and root scopes
for $1\le n\le N$, all of arity at most $r$.  The unknowns are the local predicates, the
root predicates, and canonical annotations $\gamma_w$ for all $w\in\Sigma^n$,
$n\le N$.

\begin{lemma}[Finite template feasibility]
\label{lem:finite}
For fixed parameters $(\Sigma,M,\eta,F,r,\Lambda,s,N)$ and fixed scopes, the existence
of a perfectly sound canonical annotation family with stability at most $s$ for all
lengths $\le N$ is decidable by a finite constraint system.
\end{lemma}

\begin{proof}
There are finitely many words of length at most $N$, finitely many annotation cells,
and finitely many possible labels.  Thus the values $\gamma_w(j)$ are finite variables.
Each local predicate is a finite truth table over at most $r$ labels and $r$ input
symbols; each root predicate is another finite truth table.  Completeness and soundness
are finite Boolean constraints: for every word $w$, the conjunction of the selected
truth--table entries on $(w,\gamma_w)$ must equal the truth value of $w\in L$, and for
every nonmember $w$ and every possible annotation $\gamma$ the same conjunction must be
false.  The stability bound adds, for every substitution $w\to w'$, a cardinality
constraint
\[
  \left|\{j:\gamma_w(j)\ne\gamma_{w'}(j)\}\right|\le s .
\]
All quantification is over finite sets, so the resulting feasibility question is a
finite decision problem.
\end{proof}

\begin{remark}[Use of the finite formulation]
The lemma gives a disciplined way to test proposed lower--bound statements before
attempting an asymptotic proof.  For example, one can fix parity, arity $r$, stability
$s=1$ or $2$, and increasingly large $N$, then ask whether any bounded--scope template
survives.  Unsatisfiable finite instances do not by themselves prove the asymptotic
statement, but they can expose the minimal forbidden patterns that a proof must
explain.
\end{remark}

\paragraph{A concrete encoding.}
For fixed scopes, the feasibility problem can be written as a finite Boolean instance.
Use variables $X_{w,j,\lambda}$ to state that the canonical label of cell $j$ on word
$w$ is $\lambda$, and add exactly--one constraints over $\lambda\in\Lambda$.  A local
predicate table contributes variables for each possible tuple of read symbols and
labels.  Completeness and soundness are then direct implications from the selected
table entries.  Stability constraints compare the variables for adjacent words in the
substitution graph on $\Sigma^n$ and enforce a Hamming bound.  The same encoding can be
made integer-valued by replacing truth-table variables with $0$--$1$ variables and
using linear cardinality constraints.

\paragraph{Path-scope CP-SAT instances.}
The ancillary CP-SAT implementation fixes one useful scope family and leaves the
predicates arbitrary.  There are cells $0,\dots,n$, a start predicate on cell $0$, a
uniform transition predicate on triples
\[
  (\gamma(i),w_{i+1},\gamma(i+1)),
\]
and a root predicate on cell $n$.  This is the proof-labeling version of an NFA run:
the checker is nondeterministic through the certificate, not restricted to a
deterministic transition function.  For every member word, the canonical labels must
satisfy the selected predicates.  For every nonmember word, every possible label
sequence is explicitly forbidden.  Stability is imposed by cardinality constraints on
the $X_{w,j,\lambda}$ variables over substitution edges.  Each member word, nonmember
word, and substitution edge is guarded by an assumption literal; when the instance is
unsatisfiable, the solver returns a small forbidden pattern of words and edits.

The following table gives the resulting small-grid evidence.  For parity, budget
$s=1$ fails already at $n=4$ for $q=2,3,4$, while $s=2$ is feasible at $n=4,5$ for the
same label alphabets.  Thus, in this path-scope family, the first obstruction is the
stability budget rather than the number of labels.  For the locally-testable but
non-strictly-local language ``contains $11$'', the boundary is more delicate: at
$n=5$, $q=3,4$ become feasible when the budget is raised from $1$ to $2$, whereas at
$n=6$ all tested cases with $q\le4$ and $s\le2$ are infeasible.  These finite failures
do not settle Problem~\ref{prob:boundary}, but they test the locally testable
case rather than a product annotation for it.

\begin{center}
\centering
\scriptsize
\begin{tabular}{@{}llrrrr@{}}
\toprule
Language & Lengths & $q$ & $s=1$ & $s=2$ & Core size\\
\midrule
parity & $n=4,5$ & 2,3,4 & UNSAT & SAT & 5 at $s=1$\\
contains $11$ & $n=5$ & 3,4 & UNSAT & SAT & 5 at $s=1$\\
contains $11$ & $n=6$ & 3 & UNSAT & UNSAT & 5 at $s=1,2$\\
contains $11$ & $n=6$ & 4 & UNSAT & UNSAT & 53/62\\
\bottomrule
\end{tabular}
\par\smallskip
\begin{minipage}{0.88\linewidth}
\scriptsize
\textbf{Table 1.}
CP-SAT feasibility grid for uniform path scopes with arbitrary transition
predicates.  SAT means a perfectly sound path-scope certificate with the stated
canonical stability budget exists at the fixed lengths.  UNSAT core sizes are
assumption cores over member-word, nonmember-word, and substitution-edge constraints;
the $q=4,n=6$ cores are sufficient but not deletion-minimal.
\end{minipage}
\end{center}

One deletion-minimal parity core at $n=4,q=4,s=1$ consists of the member words
$0000$ and $1001$, the nonmember word $0001$, and the substitution edges
$0000\leftrightarrow1000$ and $1000\leftrightarrow1001$.  One deletion-minimal core
for ``contains $11$'' at $n=5,q=3,s=1$ consists of the nonmember word $01010$, the
member words $11010$ and $01011$, and the substitution edges
$11010\leftrightarrow11011$ and $01011\leftrightarrow11011$.  These cores are finite
forbidden substitution patterns, not histograms.

\paragraph{A bounded-width ladder scope.}
The path family has one label per cut.  To test whether the obstruction is specific to
that one-track geometry, the ancillary library also includes a
bounded-width ladder encoding.  A width-$t$ instance has $t$ cells in each column and
one arbitrary local predicate on
\[
  (\gamma_i^{(1)},\dots,\gamma_i^{(t)},\; w_{i+1},\;
    \gamma_{i+1}^{(1)},\dots,\gamma_{i+1}^{(t)}).
\]
Nonmember soundness is still exact: every assignment to all ladder cells of a
nonmember word is enumerated and rejected.  Width $1$ is the path-scope instance above;
width $2$ is a bounded-degree scope hypergraph rather than a single run.

For the language ``contains $11$'', width $2$ with binary cell labels reproduces the
same budget boundary at $n=5$: $s=1$ is UNSAT with a deletion-minimal five-assumption
core, while $s=2$ is SAT.  At $n=6$, width $2$, $q=2$, and $s=2$ is already UNSAT with a
sufficient core of size $64$.  For parity, the same width-$2$, $q=2$, $s=2$ setting is
SAT at $n=5$ and UNSAT at $n=6$ with a sufficient core of size $44$.  These are still
finite scope-family refutations rather than arbitrary-scheme lower bounds, but they
show that the obstruction is not confined to the single-cell path encoding.

\paragraph{A bounded-span hypergraph scope.}
The adjacent ladder can itself be too restrictive.  A broader constant-degree family
lets one local predicate read column $i$, the two-bit block $w_iw_{i+1}$, and column
$i+2$.  This span-$2$ scope is still finite and exactly sound against every
annotation of every nonmember word, but it is no longer a single path or an adjacent
ladder transition.

For ``contains $11$'', the width-$2$, span-$2$, $q=2$ instance with stability budget
$s=2$ is SAT at $n=6$, whereas the adjacent ladder at the same width and budget is
UNSAT.  Thus one finite obstruction can disappear when the scope hypergraph is widened.
At the next tested length the obstruction returns: the span-$2$ instance is UNSAT at
$n=7$ with a sufficient assumption core of size $45$.  For parity, the same span-$2$,
width-$2$, $q=2$ family is SAT at $n=6$ even for $s=1$.  These certificates show that
path-like finite failures must be tested against bounded-scope enrichments before they
are used as evidence for an asymptotic obstruction.

\paragraph{A deterministic path-local submodel.}
One finite family is small enough to enumerate without a SAT solver.  Fix a label set
of size $q$ and put one annotation cell after each prefix of the word.  A template is a
deterministic transition system: an initial label, one transition function for each
input symbol, and a root accepting set.  Local rules force the unique run
\[
  s_0,\ s_{i+1}=\delta_{w_{i+1}}(s_i),
\]
and the root rule reads $s_n$.  For binary input there are $q^{2q}\cdot q\cdot 2^q$
such templates.  This is a strict submodel of Lemma~\ref{lem:finite}: it fixes the
path scopes, fixes the local predicate shape to deterministic transition equations,
and therefore cannot refute arbitrary bounded-arity annotations.  It is nevertheless a
useful sanity test for the prefix-like explanations that naturally arise from DFA runs.
The A100 certificates in Section~\ref{sec:artifact} exhaust the case $q=4$, $n=13$ for
parity and for the language $|w|_1\equiv0\bmod 3$.  For parity, among
$4{,}194{,}304$ templates, $20{,}856$ recognise the language on every length up to
$13$, and every exact template has worst-case run-annotation movement $13$.  For the
modulo-$3$ language, $1{,}008$ templates are exact up to length $13$, and again every
exact template has worst-case movement $13$.

\paragraph{Section-optimal movement for fixed path recognisers.}
The deterministic and nondeterministic CUDA searches use a specified canonical run.
To separate the choice of checker from the choice of section, a second CP-SAT model
fixes the standard deterministic path recogniser and optimises the canonical annotation
map itself.  For member words the selected labels must form an accepting local run.
For nonmember words the canonical labels are unconstrained, exactly as permitted by
Definition~\ref{def:costs}.  The solver then asks for the minimum $s$ such that every
substitution edge has Hamming movement at most $s$.

\begin{center}
\centering
\scriptsize
\begin{tabular}{@{}lrrr@{}}
\toprule
Fixed path recogniser & $n=8$ & $n=10$ & $n=12$\\
\midrule
parity & 4 & 5 & 6\\
$|w|_1=0\bmod 3$ & 4 & 5 & 6\\
contains $11$ & 7 & 9 & 11\\
\bottomrule
\end{tabular}
\par\smallskip
\begin{minipage}{0.88\linewidth}
\scriptsize
\textbf{Table 2.}
Exact minimum worst-case movement over canonical sections of the fixed standard
deterministic path recogniser.  The checker and scopes are fixed; the section is
optimised.  Nonmember canonical annotations are unrestricted.  These values are not
optima over all local annotation schemes.
\end{minipage}
\end{center}

The table changes the interpretation of the CUDA run statistics.  For the group
languages, a lexicographic or unique locally consistent run can move on all $n$ path
cells, while the section optimum for the fixed recogniser is only $n/2$ at the tested
even lengths.  For ``contains $11$'', the same optimisation still gives $n-1$.  Thus
canonical-section choice is a real degree of freedom, but in this fixed path geometry
the locally-testable example exhibits stronger finite transport than the two group
examples.  No asymptotic conclusion follows without analysing other recognisers and
scope families.

\paragraph{Why finite failures matter.}
A finite unsatisfiable instance is not merely a failed search over one hand-designed
scheme.  It rules out every scheme with the chosen scopes, label alphabet, and stability
budget up to length $N$.  When an unsatisfiable core is small, it identifies a finite
pattern of words and substitutions that any asymptotic lower bound must generalise.
This is particularly useful for separating two possible explanations: failure caused by
too small a label alphabet, and failure caused by the stability budget itself.

\paragraph{Scope enumeration.}
The remaining nonuniformity lies in the choice of scopes.  For bounded arity and fixed
$N$ there are still finitely many choices.  One can therefore enumerate scope hypergraphs
up to isomorphism, or restrict attention to translation--invariant and bounded--degree
families.  The first route is exhaustive but large; the second is less general but
closer to the product and window geometries studied above.  In either case the problem
is finite at every scale.

\paragraph{Monotone obstruction sequences.}
For fixed template shape family $(I(n),\Pi(n))_{n\le N}$, infeasibility is monotone in
three parameters: decreasing the label alphabet, decreasing the stability budget, or
adding more words to the tested length range cannot restore feasibility.  Thus a chain
of unsatisfiable instances at increasing $N$ is not a collection of unrelated negative
experiments; it is evidence for a single obstruction pattern.  The natural proof task is
to identify a bounded set of substitutions whose induced constraints force a path-like
transport of a group or counting value.  Proposition~\ref{prop:grouplower} proves this
path phenomenon for product geometries; the finite search asks whether every exact
constant-stability scheme for an edit-active group quotient must contain such a path in
some encoded form.

\section{Ancillary library}
\label{sec:artifact}

The finite search above is useful only if its outputs are inspectable.  The
ancillary material is therefore organised by the claims it supports.  The Lean~4
branch \cite{deMouraUllrich21}, under \texttt{anc/lean}, has separate modules for the
Boolean presentation of $\mathbb Z/2\mathbb Z$, list-level parity products, and binary
product trees.  The central checked lemma is that flipping a symbol inside any context
flips the product label of that interval.  This is the kernel of the parity instance of
Proposition~\ref{prop:grouplower}: every ancestor interval containing the edited leaf
has a different product label.  The development is independent of
Mathlib, so the trusted surface is the Lean kernel and the pinned Lean~4 toolchain.

\begin{center}
\centering
\scriptsize
\begin{tabular*}{\linewidth}{@{\extracolsep{\fill}}lll@{}}
\toprule
Claim tested & Certificate entry point & Claim strength\\
\midrule
Parity interval products flip under edit & \texttt{anc/lean} & checked proof kernel\\
Product mean separates group/aperiodic rows & \texttt{anc/gpu/materialized} & exact finite profile\\
Path-scope finite feasibility and cores & \texttt{anc/sat/path} & fixed-scope CP-SAT\\
Width-$2$ ladder obstruction & \texttt{anc/sat/ladder} & bounded-degree CP-SAT\\
Span-$2$ hypergraph scope & \texttt{anc/sat/span2} & broader-scope CP-SAT\\
Optimal sections of fixed path checkers & \texttt{anc/sat/section\_optimal} & exact section optimum\\
Deterministic run templates & \texttt{anc/gpu/deterministic\_path} & restricted exhaustive search\\
Nondeterministic path relations & \texttt{anc/gpu/nfa\_path\_relation} & relation-template search\\
Abelian/non-abelian prefix runs & \texttt{anc/gpu/group\_prefix} & canonical-run audit\\
Context-free interval charts & \texttt{anc/gpu/cfg\_chart} & nonterminal/split profile\\
\bottomrule
\end{tabular*}
\par\smallskip
\begin{minipage}{0.88\linewidth}
\scriptsize
\textbf{Certificate map.}
Each entry is tied to a bounded claim.  The CP-SAT and CUDA searches are finite
certificates for named template families; they are not advertised as asymptotic lower
bounds for arbitrary bounded-arity schemes.
\end{minipage}
\end{center}

The SAT branch, under \texttt{anc/sat}, contains the CP-SAT implementations used for
Table~1, the ladder and span-$2$ hypergraph scope tests, and the JSON certificates for
the assumption cores.  This is the branch closest to Lemma~\ref{lem:finite}: the local
predicates are variables, the canonical labels are variables, nonmember soundness is
universally enumerated over all annotations, and stability is a cardinality constraint
over substitution edges.

The GPU branch, under \texttt{anc/gpu}, has six pipelines.  The product-atlas pipeline
materialises batches of canonical balanced-product annotations in VRAM and scans the
substitution edges induced by those annotations.  This is retained as a deterministic
regression test and as an exact distributional profile for
Proposition~\ref{prop:meanmovement}, not as an experiment that could change
Theorem~\ref{thm:universal}.  A lighter on-the-fly edge-profile scanner is kept as a
cross-check.  The deterministic path-local template search materialises the final-state
table for every template-word pair before filtering exact templates for parity or
modulo $3$.  The nondeterministic path-relation search enumerates arbitrary start
masks, transition relations, and root masks over a small label alphabet; exactness is
checked on every word up to the target length, and movement is then measured for the
lexicographically least accepted local run.  The group-prefix pipeline materialises
canonical prefix runs for $\mathbb Z/3\mathbb Z$ and $S_3$, testing whether the path
movement mechanism depends on commutativity.  The context-free chart pipeline
materialises interval nonterminal charts and canonical accepting-split witnesses for
fixed grammars, then scans terminal substitutions.  It is not a theorem about all
context-free certificates; it is a finite certificate showing that the same search
method can record grammatical relabeling and derivation-choice movement.

The included A100 materialised product certificate was produced on one NVIDIA A100-SXM4
40GB GPU.  For length $n=29$ it materialises batches of $56{,}512{,}727$ words, using
$38{,}654{,}705{,}268$ bytes for the annotation slab.  Across all batches it scans
\[
  2^{29}\cdot 29 = 15{,}569{,}256{,}448
\]
directed one-bit substitution edges.  The materialisation kernels took $9.84$ seconds
in total and the scan kernels took $49.42$ seconds.  The JSON certificate stores the
histograms, root-membership flips, timings, device metadata, memory allocation size,
and 64-bit digests; it does not store the per-edge trace, because the
histograms are the finite object being analysed.

\begin{center}
\centering
\scriptsize
\begin{tabular}{@{}llrrr@{}}
\toprule
Language & Hint & Root flips & Max & Mean\\
\midrule
even number of $1$s & group & 15{,}569{,}256{,}448 & 6 & 5.8966\\
number of $1$s $=0\bmod 3$ & group & 10{,}379{,}504{,}318 & 6 & 5.8966\\
contains $11$ & LT & 34{,}097{,}282 & 6 & 2.8286\\
contains $1$ & LT & 58 & 6 & 1.6256\\
ends with $1$ & SL suffix & 536{,}870{,}912 & 6 & 1.9655\\
avoids $11$ & SL factor & 34{,}097{,}282 & 6 & 2.8286\\
\bottomrule
\end{tabular}
\par\smallskip
\begin{minipage}{0.88\linewidth}
\scriptsize
\textbf{Table 3.} Selected product-annotation profiles from the A100
materialised certificate at $n=29$.  The Max and Mean columns count changed
canonical balanced-product labels under directed one-symbol substitutions; they do
not count the best possible annotation for the language.  Thus strictly local rows can
still have nonzero product changes, while Proposition~\ref{prop:slt} gives
annotation-free schemes with zero annotation stability cost.
\end{minipage}
\end{center}

The table should not be read as evidence that the product annotation is optimal.  Its
maximum is fixed by the tree height, and the group rows necessarily saturate that
height.  Its value is narrower: the mean column shows that worst-case and distributional
movement can diverge.  Aperiodic examples have much smaller average product
movement, and strictly local examples show the modelling distinction most clearly:
the product annotation is not optimal for them, because the annotation-free window
recogniser already has $\edit(n)=0$.

The deterministic path-template certificates are smaller on disk but closer to the
finite-feasibility formulation.  For $q=4$ and $n=13$ each run enumerates every binary
deterministic path-local template with four labels.  Each run materialises
$34{,}359{,}738{,}368$ bytes of final-state data in A100 VRAM.  The parity run has
$20{,}856$ exact templates; the modulo-$3$ run has $1{,}008$.  In both cases every
exact template has worst-case unique-run movement $13$ at length $13$.  The finite
outcome is summarised in Table~4.

\begin{center}
\centering
\scriptsize
\begin{tabular}{@{}llrrr@{}}
\toprule
Target & Family & $q$ & Exact templates & Minimum worst movement\\
\midrule
parity & deterministic path-local & 4 & 20{,}856 & 13\\
$|w|_1=0\bmod 3$ & deterministic path-local & 4 & 1{,}008 & 13\\
\bottomrule
\end{tabular}
\par\smallskip
\begin{minipage}{0.88\linewidth}
\scriptsize
\textbf{Table 4.} A100 finite-template certificates for deterministic path-local
annotations at $n=13$.  The movement column counts changed cells in the unique run
annotation under one-bit substitutions.  This is evidence for the prefix-run
obstruction in this restricted family, not a lower bound for arbitrary bounded-arity
annotation schemes.
\end{minipage}
\end{center}

The nondeterministic path-relation certificate is closer to
Definition~\ref{def:scheme} than the deterministic run search because the local
transition predicate is an arbitrary relation rather than a function.  For $q=3$ the
template has $2q$ bits for the start and root masks and $2q^2$ bits for the two
letter-indexed transition relations, hence $2^{24}=16{,}777{,}216$ candidate local
predicates.  The A100 run materialises a 32\,GiB acceptance table per batch.  For
$n=11$ it finds $2{,}442$ exact parity relations, $6$ exact modulo-$3$ relations, and
$1{,}632$ exact ``contains $11$'' relations.  The same exact relation counts persist at
$n=12$, $n=14$, $n=15$, and $n=16$.  For the chosen lexicographic canonical section, the
minimum worst movement equals the tested length in all three target languages: $11$,
$12$, $14$, $15$, and $16$, respectively.  The $n=16$ run uses $32$ batches and
$891.6$ seconds of evaluation-kernel time.  This is not an optimal-section lower bound:
it is a finite certificate showing that even nondeterministic path-local
predicates have no low-movement lexicographic section on these lengths.

Finally, the group-prefix pipeline materialises the canonical prefix runs for
$\mathbb Z/3\mathbb Z$ and for the non-abelian group $S_3$ at $n=29$.  It uses
$32{,}212{,}254{,}720$ bytes for the two prefix tables and scans
$15{,}569{,}256{,}448$ directed substitutions.  Both groups have maximum movement
$29$ and mean movement $15$.  The identical histogram is consistent with the proof
mechanism in Proposition~\ref{prop:grouplower}: cancellation, not commutativity, is the
reason a changed generator moves every later prefix value.  This is still a canonical
run audit rather than an exhaustive search over arbitrary $S_3$-recognising templates.

The context-free chart certificate is an experiment beyond regular monoids.  It uses
three small Chomsky-normal-form style grammars: nonempty binary palindromes,
$\{0^k1^k:k\ge1\}$, and one-type Dyck parentheses with concatenation.  For every
binary word of length $24$ it materialises the full interval chart.  Each chart cell
stores two bytes: the generated nonterminal-set bitmask and the first split witness for
the accepting nonterminal, or $255$ when that nonterminal is absent.  The A100 run
materialises $30{,}198{,}988{,}800$ bytes and scans $402{,}653{,}184$ directed
substitutions.  The finite outcome is summarised in Table~5.

\begin{center}
\centering
\scriptsize
\begin{tabular}{@{}lrrrr@{}}
\toprule
Grammar & Root flips & Max chart & Mean chart & Mean split\\
\midrule
palindrome & 196{,}608 & 156 & 22.4270 & 12.3376\\
$0^k1^k$ & 48 & 46 & 3.6667 & 1.6049\\
Dyck, one type & 9{,}984{,}576 & 55 & 9.0325 & 5.1549\\
\bottomrule
\end{tabular}
\par\smallskip
\begin{minipage}{0.88\linewidth}
\scriptsize
\textbf{Table 5.} Context-free interval-chart movement at $n=24$.  ``Max chart'' and
``Mean chart'' count changed interval nonterminal-set cells.  ``Mean split'' counts
changed canonical accepting-split witnesses and is a finite surrogate for derivation
choice movement under the fixed interval universe.  It is not an optimal parse-tree
stability theorem.
\end{minipage}
\end{center}

All A100 certificates record compact digests rather than raw traces.  The recorded
environment for the new runs is NVIDIA driver 595.58.03, CUDA toolkit 13.2.78,
NVIDIA A100-SXM4 40GB, persistence mode enabled, compute mode default, MIG disabled,
and zero volatile uncorrected ECC errors.  CUDA reductions use fixed block counts and
integer atomic additions/xors, so the aggregate histograms and 64-bit digests are
deterministic for the given kernel and parameters.

\section{Insertion and deletion edits}
\label{sec:indel}

The paper treats substitutions because the cell set is then fixed and edit stability is
a literal Hamming distance between annotations.  Insertions and deletions require an
alignment convention: one must say which old cells correspond to which new cells.
With persistent node identities in a balanced search tree keyed by position, the same
monoid labels can be maintained with logarithmic amortised structural change, but the
proper definition of stability must count both relabelled cells and rebalanced cells.
This is a useful extension, but it is a different model from the fixed--length one used
in the theorems above.

\section{Context--free outlook}
\label{sec:cfl}

The regular case works because recognition factors through a finite associative
product.  For context--free languages, the analogous object is a derivation tree rather
than a monoid product.  A plausible exact annotation would label a balanced parse
decomposition by nonterminals and local production checks.  The corresponding stability
question is then: after one terminal substitution, how much of a balanced derivation
must be relabelled before all local production rules hold again?  This connects the
present model to dynamic parsing, but the controlling invariant is grammatical rather
than monoidal.

The ancillary CFG chart pipeline is the first finite version of this question.  It
does not attempt to solve context-free certification.  Instead, for a fixed grammar it
materialises the exact interval chart
\[
  C_w[i,j]\subseteq N
\]
of nonterminals deriving the factor $w_i\cdots w_{j-1}$, together with a canonical
first split witnessing the accepting nonterminal when such a split exists.  The local
production rules are the usual CKY implications: leaf cells are forced by terminal
rules, and an interval cell is justified by a binary production and a split into two
smaller intervals.  Under a terminal substitution the experiment measures two finite
movements: nonterminal-set relabeling and accepting-split witness movement.

This chart is larger than a single parse tree.  It avoids depending on
an arbitrary parser's choice of one derivation, and it gives every substring a
stable cell identity in the fixed-length substitution model.  The split-witness column
in Table~5 measures derivation choice: it records how often the canonical local
decomposition of an interval changes even when the generated nonterminal set may stay
the same.  For insertions and deletions, one must add the aligned cell metric from
Definition~\ref{def:alignededit}, because interval identities and balanced parse
tree nodes no longer have the same domain before and after the edit.

A pushdown-run formulation is possible but less stable.  The annotation would record
input position, stack action, and control transition.  One terminal substitution can
alter a long suffix of a naive accepting computation, so the immediate analogue of the
regular prefix run inherits the wrong geometry.  Balanced derivations or interval
charts are therefore more appropriate finite objects.  The next context-free step is
a finite feasibility encoding whose variables choose
local production predicates, canonical nonterminal labels, and aligned split witnesses
under a bounded stability budget, mirroring Lemma~\ref{lem:finite} for grammars rather
than monoids.

\section{Scope of the results}
\label{sec:scope}

The proved results are Lemma~\ref{lem:quotientannotation}, Lemma~\ref{lem:cellbit},
Lemma~\ref{lem:syntacticvisible}, Theorem~\ref{thm:universal}, Proposition~\ref{prop:meanmovement},
Proposition~\ref{prop:dary}, Proposition~\ref{prop:grouplower}, Corollary~\ref{prop:productlower},
Lemma~\ref{lem:prefix}, Theorem~\ref{thm:separation}, Proposition~\ref{prop:slt},
Proposition~\ref{prop:gap}, Proposition~\ref{prop:dynfodictionary},
Proposition~\ref{prop:twosidedclosure},
Lemma~\ref{lem:finite}, and Proposition~\ref{prop:alignedproduct}.
Problem~\ref{prob:boundary} and Problem~\ref{conj:dynfo} are not used as assumptions
in any proof.  The paper makes no cryptographic claim and no sublinear exact--checking
claim.  Its main resource is local deterministic revalidation under a fixed input
length; Proposition~\ref{prop:alignedproduct} is a separate aligned edit-trace
extension.  The Lean branch checks a small formal kernel of the parity product
argument; it does not formalise every theorem in the paper.  The SAT and GPU branches
are finite, reproducible sources of test instances and obstruction cores, not
asymptotic lower bounds for arbitrary annotation schemes.

Two modelling choices are worth stressing.  First, the update operation is a
substitution, so the set of input positions is unchanged.  This lets us count changed
annotation cells without introducing an alignment metric.  Second, the verifier is
external to the procedure that constructs the annotation.  A trusted dynamic algorithm
may keep information that an external verifier cannot locally force; the parity example
is designed precisely to show that these two notions should not be conflated.

\section{Open problems and technical obstacles}
\label{sec:future}

Let $\mathcal C_{\mathrm{const}}$ be the class of regular languages admitting a
perfectly sound bounded--arity scheme with $\edit(n)=\bigO(1)$ and
$\rev(n)=\bigO(1)$, and let $\mathcal C_{\log}$ be the analogous class with
logarithmic bounds.  The results of the paper locate the problem as
\[
  \St\subseteq \mathcal C_{\mathrm{const}}
  \subseteq \mathrm{Reg}\subseteq \mathcal C_{\log}.
\]
The left inclusion is Proposition~\ref{prop:slt}; the right inclusion is
Theorem~\ref{thm:universal}.  The missing information is not whether regular languages
have local exact explanations: they do, with logarithmic movement.  The missing
information is whether local soundness itself forces long-range transport for some
regular languages, even when the annotation is allowed to abandon the product tree.

\paragraph{The constant-stability boundary.}
The main open question is Problem~\ref{prob:boundary}: is
$\mathcal C_{\mathrm{const}}$ exactly $\St$, and if not, what is the first
non-strict-local language it contains?  Before this question can be converted into an
algebraic variety conjecture one must prove closure of the relevant resource class
under the relevant language operations.  Proposition~\ref{prop:twosidedclosure} gives
Boolean closure, word-quotient closure, finite product closure at the level of
decision schemes, uniform inverse-morphism closure in the substitution model, and
arbitrary fixed inverse-morphism closure in the aligned edit-trace model.  It does not
close the original one-sided existential class: nonmembers may have no accepted
annotations, and the original fixed-length metric still cannot see the insertions and
deletions created by a non-uniform morphism.  The positive side is understood only at
the strict-local endpoint.  A strictly local
language is checked by bounded windows, so no informative annotation is needed.  The
first unresolved aperiodic examples are locally testable but not strictly local, for
instance languages asserting that some fixed factor occurs.  Such a language has a
finite aperiodic syntactic monoid, so the group-cancellation argument cannot apply.
However, a naive local certificate of occurrence uses a marked witness, and when the
unique occurrence is destroyed the mark may have to move a long distance or be replaced
by a distributed absence certificate.  Thus the obstruction, if one exists, must be a
combinatorial obstruction to locally forced marking rather than a group obstruction.

\paragraph{From product lower bounds to arbitrary annotations.}
Proposition~\ref{prop:grouplower} proves the logarithmic lower bound in the setting where
the algebra gives cancellation: every ancestor label is a projected product in a group.
This is why the proof stops at quotient-exposing product decompositions.  In an
arbitrary bounded--arity scheme, a cell need not have a distinguished meaning as a
prefix, interval, or quotient value; several cells may jointly encode the relevant
monoid information, and multiple accepted annotations may exist for the same word.  A
general lower bound for parity would therefore need an extraction theorem: from any
perfectly sound constant-stability checker, derive a stable local object that still
transports the $\mathbb Z/2\mathbb Z$ value across substitutions.  Existing trusted
dynamic-membership classifications explain when a maintained auxiliary state can be
updated quickly, but they do not give such an externally checkable extraction.  This
extraction problem is the main mathematical obstacle.

\paragraph{Trusted dynamic maintenance.}
The constant-profile class should also be compared with trusted dynamic-maintenance
classes for regular languages, but only after the external-checkability requirement is
kept explicit.  A trusted update program may maintain auxiliary relations that have no
short local certificate after the update.  Conversely, a constant-profile annotation
gives a locally certified maintained representation, but the paper does not prove that
every trusted dynamic program can be converted into one.

\begin{problem}[External certification versus trusted update]
\label{conj:dynfo}
Compare the regular languages admitting bounded $\loc$, bounded $\edit$, and bounded
$\rev$ with the regular-language classes maintainable by quantifier-free or
first-order dynamic updates.  Does either containment hold?  Any strict-separation
claim must account for the fact that trusted dynamic state need not be externally
certifiable, while a canonical local section need not be generated by a first-order
update program.
\end{problem}

The following elementary implication is the positive part currently available.

\begin{proposition}[A definable selector gives a trusted dynamic program]
\label{prop:dynfodictionary}
Suppose a two-sided total annotation scheme has a uniform first-order presentation of
its cell universe, label predicates, local-rule incidence relation, and canonical
update selector: after a substitution, first-order formulas over the old word, the old
annotation, the edited position, and the new letter identify exactly the cells whose
labels change and their new labels.  Then the recognised language is maintainable in
DynFO.  If the presentation and selector are quantifier-free, the maintenance program
is quantifier-free.
\end{proposition}

\begin{proof}
Store one auxiliary relation for each label value and each finite cell sort in the
presentation, together with the claimed output bit.  The selector formulas update the
label relations after the edited input symbol changes.  The root predicate of the
two-sided scheme updates the output bit, and exactness of the scheme guarantees that
this bit is the membership value.  When the selector formulas are quantifier-free, no
quantifiers are introduced in these auxiliary-relation updates.
\end{proof}

This proposition should not be read backwards.  DynFO supplies trusted update formulas
for some auxiliary relations, but it need not supply a bounded-arity local checker for
those relations; a local checker supplies external soundness, but it need not supply a
first-order definable way to choose the next accepted representative.

\paragraph{Finite obstruction search as proof discovery.}
Lemma~\ref{lem:finite} gives a finite CSP for every fixed template, label alphabet,
stability budget, and length bound.  The CP-SAT path-scope certificates are the first
implementation of the full variable pattern $X_{w,j,\lambda}$ with arbitrary local
predicate truth tables and exact nonmember soundness.  They already produce small
unsatisfiable cores for parity and for the locally-testable language ``contains $11$''.
The CUDA deterministic-template certificates are a separate, more scalable but more
restricted family: they fix prefix-path scopes and deterministic transition rules, then
exhaust $4{,}194{,}304$ four-state templates up to length $13$.  The next technical step
is to move beyond the current path and width-$2$ ladder families: enumerate
bounded-degree scope hypergraphs up to isomorphism, add symmetry reduction from
available position symmetries and quotient actions of the recognising monoid, and
extract smaller cores.  The hard part remains soundness: for every nonmember word, all
possible annotations must be rejected.  That universal quantification is what turns the
finite search from a template-fitting task into a proof-complexity problem.

\paragraph{Additional resources inside the fixed-length model.}
The present stability metric counts changed cells.  This is appropriate for labelled
local data structures, but it hides two secondary costs that matter for lower bounds.
The main model fixes the alphabet, so Lemma~\ref{lem:cellbit} makes cell and bit
movement equivalent up to a constant.  In a variable-alphabet extension, however, a
scheme may pack several algebraic values into one cell; then the bit cost in
$\size(n)$ and $\bedit(n)$ can diverge from the cell cost in $\edit(n)$.  A second
issue remains even in the fixed-alphabet model: a small number of changed labels need
not imply that the updated canonical annotation can be found locally.  A more detailed theory
should relate changed cells, changed bits, affected rules, and search time for finding
the next accepted representative.  Product trees make all four costs logarithmic at
once; a non-product scheme that improves one of them may worsen another.

\paragraph{Aligned insertions and deletions.}
Insertions and deletions are not just longer substitutions: the annotation cell sets
have different domains before and after the edit.  The correct object is therefore a
metric on aligned annotations.

\begin{definition}[Aligned string-edit movement]
\label{def:alignededit}
An aligned edit edge from $w$ to $w'$ is a Levenshtein edit together with the monotone
partial bijection between the surviving input positions.  If annotations live on cell
sets $I$ and $I'$, an annotation alignment is a partial bijection
$\alpha:I\rightharpoonup I'$.  Its movement cost is
\[
  |I\setminus\mathrm{dom}(\alpha)|+
  |I'\setminus\mathrm{im}(\alpha)|+
  |\{i\in\mathrm{dom}(\alpha):\gamma(i)\ne\gamma'(\alpha(i))\}|.
\]
The aligned edit stability of a construction is the maximum such cost over one
substitution, insertion, or deletion edge, using the alignment supplied by the edit
trace.
\end{definition}

\begin{proposition}[Balanced products under aligned string edits]
\label{prop:alignedproduct}
In the aligned edit-trace model, every regular language has an exact balanced-product
annotation with linear size, constant locality, and $\bigO(\log n)$ aligned movement
and revalidation under a single substitution, insertion, or deletion.
\end{proposition}

\begin{proof}
Maintain the leaves in their edit-trace order and store monoid products in a
deterministic weight-balanced binary rope.  A substitution changes one leaf label and
therefore only the labels on its root path.  An insertion splices in one new leaf, and
a deletion removes one leaf.  Standard weight-balanced rebalancing performs only
$\bigO(\log n)$ rotations and touches only $\bigO(\log n)$ nodes on the search and
rotation paths.  All off-path subtrees keep both their cell identity and their monoid
label.  Each touched internal node is still checked by the same constant-arity product
rule, so the number of changed, created, or deleted cells and the number of local rules
requiring revalidation are both logarithmic.
\end{proof}

This proposition is an edit-trace statement.  It does not settle the
plain-string question in which two equal symbols have no persistent identity and a
canonical representative must choose an alignment without being handed one by the
edit.  Non-uniform inverse morphisms fall into exactly this aligned-edit gap: changing
one input letter can replace a bounded block by a block of different length.

\paragraph{Context-free analogues.}
For context-free languages, the monoid product is replaced by a derivation or stack
object.  The chart certificate in Table~5 is the first bounded experiment in this
direction: it measures interval nonterminal relabeling and canonical split-witness
movement for fixed grammars.  The analogue of the syntactic monoid would have to turn
that finite measurement into a grammar-invariant measure of how local production
certificates change after a terminal edit.  The likely obstruction is no longer group
transport, but the instability of parse choices: a single terminal substitution can
force a different derivation even when membership itself remains true.



\begin{thebibliography}{99}

\bibitem{AJP21}
A.~Amarilli, L.~Jachiet, and C.~Paperman.
\newblock Dynamic membership for regular languages.
\newblock In \emph{48th International Colloquium on Automata, Languages, and
Programming (ICALP)}, LIPIcs 198, 116:1--116:17, 2021.

\bibitem{KKP10}
A.~Korman, S.~Kutten, and D.~Peleg.
\newblock Proof labeling schemes.
\newblock \emph{Distributed Computing}, 22(4):215--233, 2010.

\bibitem{GS16}
M.~G\"o\"os and J.~Suomela.
\newblock Locally checkable proofs in distributed computing.
\newblock \emph{Theory of Computing}, 12(19):1--33, 2016.

\bibitem{FF17}
L.~Feuilloley and P.~Fraigniaud.
\newblock Error-sensitive proof-labeling schemes.
\newblock In \emph{31st International Symposium on Distributed Computing (DISC)},
LIPIcs 91, 16:1--16:16, 2017.

\bibitem{Merkle87}
R.~C.~Merkle.
\newblock A digital signature based on a conventional encryption function.
\newblock In \emph{Advances in Cryptology (CRYPTO)}, LNCS 293, 369--378, 1987.

\bibitem{LadnerFischer80}
R.~E.~Ladner and M.~J.~Fischer.
\newblock Parallel prefix computation.
\newblock \emph{Journal of the ACM}, 27(4):831--838, 1980.

\bibitem{McNaughtonPapert71}
R.~McNaughton and S.~Papert.
\newblock \emph{Counter-Free Automata}.
\newblock MIT Press, 1971.

\bibitem{BrzozowskiSimon73}
J.~A.~Brzozowski and I.~Simon.
\newblock Characterizations of locally testable events.
\newblock \emph{Discrete Mathematics}, 4(3):243--271, 1973.

\bibitem{Zalcstein72}
Y.~Zalcstein.
\newblock Locally testable languages.
\newblock \emph{Journal of Computer and System Sciences}, 6(2):151--167, 1972.

\bibitem{Schutzenberger65}
M.~P.~Sch\"utzenberger.
\newblock On finite monoids having only trivial subgroups.
\newblock \emph{Information and Control}, 8(2):190--194, 1965.

\bibitem{Pin86}
J.-\'E.~Pin.
\newblock \emph{Varieties of Formal Languages}.
\newblock Plenum Press, 1986.

\bibitem{Eilenberg76}
S.~Eilenberg.
\newblock \emph{Automata, Languages, and Machines}, Volume B.
\newblock Academic Press, 1976.


\bibitem{PatnaikImmerman94}
S.~Patnaik and N.~Immerman.
\newblock Dyn-FO: A parallel, dynamic complexity class.
\newblock \emph{Journal of Computer and System Sciences}, 55(2):199--209, 1997.

\bibitem{GMS12}
W.~Gelade, M.~Marquardt, and T.~Schwentick.
\newblock The dynamic complexity of formal languages.
\newblock \emph{ACM Transactions on Computational Logic}, 13(3):1--36, 2012.

\bibitem{Tschirbs23}
F.~Tschirbs.
\newblock Dynamic complexity of regular languages.
\newblock In \emph{31st EACSL Annual Conference on Computer Science Logic (CSL)},
LIPIcs 252, 35:1--35:20, 2023.

\bibitem{BTVZ26}
C.~Barloy, F.~Tschirbs, N.~Vortmeier, and T.~Zeume.
\newblock Algebraic characterizations of classes of regular languages in DynFO.
\newblock In \emph{43rd International Symposium on Theoretical Aspects of
Computer Science (STACS)}, LIPIcs 364, 9:1--9:19, 2026.

\bibitem{deMouraUllrich21}
L.~de Moura and S.~Ullrich.
\newblock The Lean 4 theorem prover and programming language.
\newblock In \emph{Automated Deduction -- CADE 28}, LNCS 12699, 625--635,
2021.

\bibitem{Cascarano10}
N.~Cascarano, P.~Rolando, F.~Risso, and R.~Sisto.
\newblock iNFAnt: NFA pattern matching on GPGPU devices.
\newblock \emph{ACM SIGCOMM Computer Communication Review}, 40(5):20--26,
2010.

\bibitem{YuBecchi13}
X.~Yu and M.~Becchi.
\newblock Exploring different automata representations for efficient regular
expression matching on GPUs.
\newblock In \emph{18th ACM SIGPLAN Symposium on Principles and Practice of
Parallel Programming (PPoPP)}, 2013.

\bibitem{Mytkowicz14}
T.~Mytkowicz, M.~Musuvathi, and W.~Schulte.
\newblock Data-parallel finite-state machines.
\newblock In \emph{19th International Conference on Architectural Support for
Programming Languages and Operating Systems (ASPLOS)}, 529--542, 2014.

\bibitem{LiuPaiJog20}
H.~Liu, S.~Pai, and A.~Jog.
\newblock Why GPUs are slow at executing NFAs and how to make them faster.
\newblock In \emph{25th International Conference on Architectural Support for
Programming Languages and Operating Systems (ASPLOS)}, 2020.

\bibitem{GeZhangLiu24}
T.~Ge, T.~Zhang, and H.~Liu.
\newblock ngAP: Non-blocking large-scale automata processing on GPUs.
\newblock In \emph{29th ACM International Conference on Architectural Support
for Programming Languages and Operating Systems (ASPLOS)}, 2024.

\bibitem{ValizadehBerger23}
M.~Valizadeh and M.~Berger.
\newblock Search-based regular expression inference on a GPU.
\newblock In \emph{44th ACM SIGPLAN Conference on Programming Language Design
and Implementation (PLDI)}, 2023.

\bibitem{HybridSA24}
A.~Le Glaunec, L.~Kong, and K.~Mamouras.
\newblock HybridSA: GPU acceleration of multi-pattern regex matching using bit
parallelism.
\newblock \emph{Proceedings of the ACM on Programming Languages},
8(OOPSLA2):1699--1728, 2024.

\end{thebibliography}
\end{document}